\documentclass[envcountsame]{llncs}
\usepackage{amsmath}
\usepackage{amssymb}
\usepackage{color}
\usepackage{ifthen}
\newcommand{\techRep}{true} %% switch here between true and false
\newcommand{\iftechrep}{\ifthenelse{\equal{\techRep}{true}}}

\newenvironment{qtheorem}[1]{%
{\mbox{}\newline\noindent\bf Theorem #1.}
\begin{itshape}%
}{%
\end{itshape}%
}
\newenvironment{qlemma}[1]{%
{\par\mbox{}\newline\noindent\bf Lemma #1.}%
\begin{itshape}%
}{%
\end{itshape}%
}

\newenvironment{qproposition}[1]{%
{\mbox{}\newline\noindent\bf Proposition #1.}
\begin{itshape}%
}{%
\end{itshape}%
}
\definecolor{orange3}{rgb}{1.0,0.2538,0.1681}
\definecolor{blau}{rgb}{0,.39608,0.74118}
\definecolor{rot}{rgb}{0.79216,.12941,0.24706}

\title{On Stabilization in Herman's Algorithm}
\author{\mbox{Stefan Kiefer\inst{1}\thanks{Stefan Kiefer is supported by a postdoctoral fellowship of the German Academic Exchange Service (DAAD).}
 \and Andrzej~Murawski\inst{2} \and Jo\"el Ouaknine\inst{1} \and James Worrell\inst{1} \and Lijun Zhang\inst{3}}}
\institute{Department of Computer Science, University of Oxford, UK
 \and Department of Computer Science, University of Leicester, UK
 \and DTU Informatics, Technical University of Denmark, Denmark
}

\begin{document}
%---------------------
\maketitle

\sloppy

\newcommand{\N}{\mathbb{N}}
\newcommand{\R}{\mathbb{R}}
\renewcommand{\P}[1]{{\cal P}\left(#1\right)}
\newcommand{\Ex}[1]{\mathbb{E}\left[#1\right]}
\newcommand{\E}{\mathbb{E}}
\newcommand{\tim}{\mathbf{T}}
\newcommand{\down}{\downarrow}
\newcommand{\up}{\uparrow}
\newcommand{\dir}[1]{\overrightarrow{#1}}
\newcommand{\gT}[1]{\tim_{#1}}
\newcommand{\nT}[2]{\overline{\tim}_{#1,#2}}
\newcommand{\dgT}[1]{\tim_{\dir{#1}}}
\newcommand{\A}[2]{A_{(#1)#2}}
\newcommand{\dA}[1]{A_{\dir{#1}}}
\newcommand{\nA}[3]{\overline{A}_{(#1)#2,#3}}
\newcommand{\vj}{\vec{j}}
\newcommand{\vy}{\vec{y}}
\newcommand{\of}{\overline{f}}
\newcommand{\yF}{y_F}
\newcommand{\yG}{y_G}
\newcommand{\F}[1]{F^{(#1)}}
\newcommand{\wF}{\widetilde{F}}%
\newcommand{\G}{G}%

\begin{abstract}
Herman's algorithm is a synchronous randomized protocol for achieving
self-stabilization in a token ring consisting of $N$ processes.  The
interaction of tokens makes the dynamics of the protocol very
difficult to analyze.  In this paper we study the expected time to
stabilization in terms of the initial configuration.

It is straightforward that the algorithm achieves stabilization almost
surely from any initial configuration, and it is known that the
worst-case expected time to stabilization (with respect to the initial
configuration) is $\Theta(N^2)$.  Our first contribution is to give an
upper bound of $0.64N^2$ on the expected stabilization time, improving
on previous upper bounds and reducing the gap with the best existing
lower bound.  We also introduce an asynchronous version of the
protocol, showing a similar $O(N^2)$ convergence bound in this case.

Assuming that errors arise from the corruption of some number $k$ of
bits, where $k$ is fixed independently of the size of the ring, we
show that the expected time to stabilization is $O(N)$.  This reveals
a hitherto unknown and highly desirable property of Herman's
algorithm: it recovers quickly from bounded errors.  We also show that
if the initial configuration arises by resetting each bit
independently and uniformly at random, then stabilization is
significantly faster than in the worst case.
\end{abstract}

%\noindent \textbf{Keywords:} Self-stabilization; Distributed Protocol;
%Probabilistic Protocol; Random Walk; Expected Time
%
%\newpage
%\pagenumbering{arabic}
%\setcounter{page}{1}

\section{Introduction} \label{sec:intro}
Self-stabilization is a concept of fault-tolerance in distributed
computing.  A system is self-stabilizing if, starting in an arbitrary
state, it reaches a correct or legitimate state and remains in a
legitimate state thereafter.  Thus a self-stabilizing system is able
to recover from \emph{transient errors} such as state-corrupting
faults.  The study of self-stabilizing algorithms originated in an
influential paper of Dijkstra~\cite{D74}.  By now there is a
considerable body of work in the area, see~\cite{Schneider,Dolev}.

In this paper we consider self-stabilization in a classical context
that was also treated in Dijkstra's original paper---a \emph{token
  ring}, i.e., a ring of $N$ identical processes, exactly one of which
is meant to hold a token at any given time.  If, through some error,
the ring enters a configuration with multiple tokens,
self-stabilization requires that the system be guaranteed to reach a
configuration with only one token.  In particular, we are interested
in analyzing a self-stabilization algorithm proposed by
Herman~\cite{Herman90}.

Herman's algorithm is a randomized procedure by which a ring of
processes connected uni-directionally can achieve self-stabilization
almost surely.  The algorithm works by having each process
synchronously execute the following action at each time step: if the
process possesses a token then it passes the token to its clockwise
neighbor with probability $1/2$ and keeps the token with probability
$1/2$.  If such a process decides to keep its token and if it receives
a token from its neighbor then the two tokens are annihilated.  Due to
the way the algorithm is implemented we can assume that an error state
always has an odd number of tokens, thus this process of pairwise
annihilation eventually leads to a configuration with a single token.

While the almost-sure termination of Herman's algorithm is
straightforward, computing the time to termination is a challenging
problem.  This is characteristic of systems of interacting particles
under random motion, which are ubiquitous in the physical and medical
sciences, including statistical mechanics, neural networks and
epidemiology~\cite{Liggett}.  The analysis of such systems typically
requires delicate combinatorial arguments~\cite{interacting}.  Our
case is no exception, and we heavily exploit work of
Balding~\cite{Balding88}, which was motivated by a scenario from
physical chemistry.

Given some initial configuration, let $\tim$ be the time until the
token ring stabilizes under Herman's algorithm.  We analyze the
expectation of $\tim$ in three natural cases: the worst case (over all
initial configurations); the case in which the initial configuration
is chosen uniformly at random; the case in which the initial
configuration arises from a legitimate configuration by a bounded
number of bit errors.  In addition we introduce and analyze an
asynchronous variant of Herman's algorithm.  The latter dispenses with
the successive time steps required in the synchronous algorithm, and
instead has each process pass its token after an exponentially
distributed time delay.

Herman's original paper~\cite{Herman90} showed that $\E \tim \leq
(N^2\log N)/2$ in the worst case (i.e., over all initial
configurations with $N$ processes).  It also mentions an improved
upper bound of $O(N^2)$ due to Dolev, Israeli, and Moran, without
giving a proof or a further reference.  In~2005, three
papers~\cite{Fribourg05,McIverMorgan,Nakata05} were published, largely
independently, all of them giving improved $O(N^2)$ bounds.  The
paper~\cite{McIverMorgan} also gives a lower bound of $4N^2/27$, which
is the expected stabilization time starting from a configuration with
three equally spaced tokens.  It was conjectured in~\cite{McIverMorgan}
that this is the worst case among all starting configurations,
including those with more than three tokens.  This intriguing
conjecture is supported by experimental evidence~\cite{PrismHerman}.

Our first result, Theorem~\ref{thm:upper-bound}, gives an upper bound
of $0.64N^2$ for the expected stabilization time in the synchronous
version of Herman's protocol (improving the constant in the hitherto
best bound by a third).  We also give an upper bound in the
asynchronous case.  To the best of our knowledge this is the first
analysis of an asynchronous version of Herman's algorithm.

To understand the other main results of the paper requires some detail
of the implementation of Herman's algorithm.  We assume that each
process has a bit that it can read and write, and that each process
can read the bit of its counterclockwise neighbor.  A process's bit
does not directly indicate the presence of a token, rather a process
has a token if it has the same bit as its counterclockwise neighbor.
Token passing is then implemented by having processes flip their bits.

In Theorem~\ref{thm:full} we provide an upper bound on the expected
time to stabilize starting from the random initial configuration, that
is, the configuration in which each process's bit is reset
independently and uniformly at random.  Herman's algorithm is such
that the random configuration is obtained in one step from the
\emph{full configuration}, i.e., the configuration in which every
process has a token.  The upper bound for the random configuration is
far better than the worst-case bound in Theorem~\ref{thm:upper-bound};
in particular, there are three-token configurations for which $\E
\tim$ is provably larger than the upper bound for the random
configuration.

In Theorem~\ref{thm:restabilization} we show that for configurations
that are obtained from a legitimate configuration by flipping a
constant number of process bits, we have $\E \tim = O(N)$; i.e., the
expected \emph{restabilization} time is linear in~$N$.  This contrasts
with the fact that there are configurations, even with only three
tokens, that need $\Omega(N^2)$ expected time for self-stabilization.
Intuitively, our result points at a highly desirable---and, to the
best of our knowledge, previously unknown--- feature of Herman's
protocol: it recovers quickly from bounded errors.  This is related to the
notion of a \emph{time adaptive protocol} from~\cite{KuttenShamir},
which refers to a protocol whose recovery time depends on the number
of state-corrupted nodes rather than the total number of nodes.

%\begin{itemize}
% \item
%  In Theorem~\ref{thm:upper-bound} we provide an upper bound on~$\E
%  \tim$ which is $O(N^2)$ and which holds for all initial
%  configurations.  It improves the previously known bounds
%  \cite{Fribourg05,McIverMorgan,Nakata05} for the synchronous case,
%  and is, to the best of our knowledge, the first analysis of the
%  asynchronous version of Herman's protocol.
% \item
%  In Theorem~\ref{thm:full} we provide an upper bound on~$\E \tim$
%  which holds for the \emph{full} configuration, i.e., the initial
%  configuration where every processor holds a token.  This upper bound
%  is far better than the one of Theorem~\ref{thm:upper-bound}; in
%  particular, there are three-token configurations for which $\E \tim$
%  is provably larger than the upper bound of Theorem~\ref{thm:full}
%  for the full configuration.
% \item
%  In Theorem~\ref{thm:restabilization} we show that for configurations
%  that are obtained from a legitimate configuration by flipping a
%  constant number of bits in the token representation, we have $\E
%  \tim = O(N)$; i.e., the expected \emph{restabilization} time is
%  linear in~$N$.  This contrasts with the fact that there are
%  configurations, even with three tokens only, which need
%  $\Omega(N^2)$ expected time for self-stabilization.  Intuitively,
%  our result points at a highly desirable---and, to the best of our
%  knowledge, previously unknown--- feature of Herman's protocol: it
%  recovers quickly from bit errors.
%\end{itemize}

Full proofs are given in
\iftechrep{the appendix.}{\cite{KMOWZ11:herman-icalp-techrep}.}

\paragraph{Related Work.}
One parameter in the design of self-stabilizing algorithms is the
number of states per machine.  In~\cite{FlateboDatta}, three different
self-stabilizing algorithms with two states per machine are
investigated.  Only one of those algorithms works in a unidirectional
ring, the other algorithms need more connections.  The ring algorithm
is probabilistic, but it is not symmetric: it requires an
``exceptional machine'' which executes different code.  Herman's
algorithm is mentioned in~\cite{FlateboDatta} as another two-state
algorithm, but it is criticized by saying ``it requires that all
machines make moves synchronously which is not easily done''.  In this
paper, we suggest and analyze an asynchronous variant of Herman's
algorithm, which is symmetric and has only two states per machine.

The protocol of~\cite{IsraeliJalfon}, also described
in~\cite{PrismHerman}, is similar to Herman's protocol in that tokens
are passed on a ring of processors.  A scheduler selects a processor
among those with a token; the selected processor passes the token to
left or right neighbor, with probability~$0.5$, respectively.  Two
colliding tokens are \emph{merged} to a single token.  Our analysis of
the asynchronous version of Herman's protocol could possibly be
adapted to this protocol, by assuming that a processor passes its
token after an exponentially distributed holding time.  Of course, the
fact that meeting tokens are merged and not annihilated would have to
be taken into account.

%For more comprehensive surveys on self-stabilization see
%\cite{Schneider,Dolev}.

\section{Preliminaries} \label{sec:prelim}

We assume $N$ processors, with $N$ odd, organized in a ring topology.
Each processor may or may not have a token.  Herman's protocol in the
traditional \emph{synchronous variant}~\cite{Herman90} works as
follows: in each time step, each processor that has a token passes its
token to its clockwise neighbor with probability~$r$ (where $0 <r <
1$ is a fixed parameter), and keeps it with probability~$1-r$; if a
processor keeps its token and receives another token from its
counterclockwise neighbor, then both of those tokens are annihilated.
Notice that the number of tokens never increases, and can decrease
only by even numbers.

Herman's protocol can be implemented as follows. Each processor
possesses a bit, which the processor can read and write.  Each
processor can also read the bit of its counterclockwise neighbor.  In
this representation having the same bit as one's counterclockwise
neighbor means having a token.  In each time step, each processor
compares its bit with the bit of its counterclockwise neighbor; if the
bits are different, the processor keeps its bit; if the bits are
equal, the processor flips its bit with probability~$r$ and keeps it
with probability~$1-r$.  It is straightforward to verify that this
procedure implements Herman's protocol: in particular a processor
flipping its bit corresponds to passing its token to its clockwise
neighbor.%
\footnote{Notice that flipping all bits in a given configuration keeps
  all tokens in place.  In fact, in the original
  formulation~\cite{Herman90}, in each iteration each bit is
  effectively flipped once more, so that flipping the bit means
  keeping the token, and keeping the bit means passing the token.  The
  two formulations are equivalent in the synchronous version, but our
  formulation allows for an asynchronous version.}

We denote the number of initial tokens by~$M$, where $1 \le M \le N$.
The token representation described above enforces that $M$ be odd.  A
configuration with only one token is called \emph{legitimate}.  The
protocol can be viewed as a Markov chain with a single bottom SCC in
which all states are legitimate configurations.  So a legitimate
configuration is reached with probability~$1$, regardless of the
initial configuration, that is, the system \emph{self-stabilizes} with
probability $1$.

In this paper we also propose and analyze an \emph{asynchronous
  variant} of Herman's protocol which works similarly to the
synchronous version.  The asynchronous variant gives rise to a
continuous-time Markov process.  Each processor with a token passes
the token to its clockwise neighbor with rate~$\lambda$, i.e., a
processor keeps its token for a time that is distributed exponentially
with parameter~$\lambda$, before passing the token to its clockwise
neighbor (i.e., flipping its bit).  The advantage of this variant is
that it does not require processor synchronization.  Note that a
processor can approximate an exponential distribution by a geometric
distribution, that is, it can execute a loop which it leaves with a
small fixed probability at each iteration.  A more precise
approximation can be obtained using a random number generator and
precise clocks.  For our performance analyses we assume an exact
exponential distribution.

Let $\tim$ denote the time until only one token is left, i.e., until
self-stabilization has occurred.  In this paper we analyze the random
variable~$\tim$, focusing mainly on its expectation~$\E \tim$.  Many
of our results hold for both the synchronous and the asynchronous
protocol version.

To aid our analysis we think of the processors as numbered from $1$
to~$N$, clockwise, according to their position in the ring.  We write
$m := (M-1)/2$.  Let $z : \{1, \ldots, M\} \to \{1, \ldots, N\}$ be
such that $z(1) < \cdots < z(M)$ and for all $i \in \{1, \ldots, M\}$,
the processor $z(i)$ initially has a token; in other words, $z(i)$ is
the position of the $i$-th token.  We often write $z_{uv}$ for $z(v) -
z(u)$.

\section{Bounds on $\E \tim$ for Arbitrary Configurations} \label{sec:morgan-nakata}

The following proposition gives a precise formula for $\E \tim$ in
both the synchronous and asynchronous protocols in case the number
of tokens is $M=3$.
\begin{proposition}[cf.~\cite{McIverMorgan}] \label{prop:triangle}
 Let $N$ denote the number of processors and let $a,b,c$ denote the
 distances between neighboring tokens, so that $a+b+c=N$.  For the
 synchronous protocol with parameter $r$ let $D=r(1-r)$, and for the
 asynchronous protocol with parameter $\lambda$ let $D=\lambda$.
 Then the expected time to stabilization is
 \[
   \E \tim = \frac{a b c}{DN} \,.
  \]
\end{proposition}
Proposition~\ref{prop:triangle} is shown in~\cite{McIverMorgan} for
the synchronous case with $r=\frac12$.  Essentially the same proof
works for $0 < r < 1$, and also in the
asynchronous case.

We call a configuration with $M=3$ equally spaced tokens an
\emph{equilateral configuration}.  If $N$ is an odd multiple of $3$
then $a = b = c = N/3$ for the equilateral configuration.  If $N$ is
not a multiple of $3$ then we ask that $a,b,c$ equal either $\lfloor
N/3 \rfloor$ or $\lceil N/3 \rceil$ .  By
Proposition~\ref{prop:triangle} the expected stabilization time for a
equilateral configuration is $\E \tim = \frac{N^2}{27D}$.  It follows
that for configurations with $M=3$ the worst case is $\E \tim =
\Omega(N^2)$ and this case arises for the equilateral configuration.
In fact it has been conjectured in~\cite{McIverMorgan} that, for
all~$N$, the equilateral configuration is the worst case, not only
among the configurations with $M=3$, but among all configurations.
This conjecture is supported by experiments carried out using the
probabilistic model checker PRISM---see~\cite{PrismHerman}.

\medskip
Finding upper bounds on $\E \tim$ in the synchronous case goes back to
Herman's original work~\cite{Herman90}.  He does not analyze $\E \tim$
in the journal version, but in his technical report~\cite{Herman90},
where he proves $\E \tim \le N^2 \lceil \log N \rceil / 2$.  He also
mentions an improvement to $O(N^2)$ due to Dolev, Israeli, and Moran,
without giving a proof or a further reference.  In~2005, three
papers~\cite{Fribourg05,McIverMorgan,Nakata05} were published, largely
independently, all of them giving improved $O(N^2)$ bounds.
In~\cite{Fribourg05} path-coupling methods are applied to
self-stabilizing protocols, which lead in the case of Herman's
protocol to the bound $\E \tim \le 2 N^2$ for the case $r=\frac12$.
Independently, the authors of~\cite{McIverMorgan} claimed $O(N^2)$.
Their proof is elementary and also shows $\E \tim \le 2 N^2$ for the
case $r=\frac12$.  Finally, the author of~\cite{Nakata05} (being aware
of the conference version of~\cite{Fribourg05}) applied the theory of
coalescing random walks to Herman's protocol to obtain $\E \tim \le
\left( \frac{\pi^2}{8} - 1 \right) \cdot \frac{N^2}{r(1-r)}$, which is
about $0.93 N^2$ for the case $r=\frac12$.  By combining results from
\cite{Nakata05} and~\cite{McIverMorgan}, we further improve the
constant in this bound (by about $32\%$), and at the same time
generalize it to the asynchronous protocol.

\newcommand{\stmtthmupperbound}{
For the
 synchronous protocol with parameter $r$ let $D=r(1-r)$, and for the
 asynchronous protocol with parameter $\lambda$ let $D=\lambda$.
 Then, for all~$N$ and for all initial configurations, we have
 \[
  \E \tim \le \left( \frac{\pi^2}{8} - \frac{29}{27} \right) \cdot \frac{N^2}{D} \,.
 \]
 Hence, $\E \tim \le 0.64 N^2$ in the synchronous case with $r=\frac12$.
}
\begin{theorem} \label{thm:upper-bound}
 \stmtthmupperbound
\end{theorem}

\section{Expressions for $\E \tim$} \label{sec:balding}

Our analysis of Herman's protocol exploits the work of
Balding~\cite{Balding88} on annihilating particle systems.  Such
systems are a special case of \emph{interacting particle systems},
which model finitely or infinitely many particles, which, in the
absence of interaction, would be modeled as independent Markov
chains.  Due to particle interaction, the evolution of a single
particle is no longer Markovian.  Interacting particle systems have
applications in many fields, including statistical mechanics, neural
networks, tumor growth and spread of infections, see \cite{Liggett}.
Balding's paper \cite{Balding88} is motivated by a scenario from
physical chemistry, where particles can be viewed as vanishing on
contact, because once two particles have met, they react and are no
longer available for reactions afterwards.  We refer the reader to
\cite{Habib01} and the references therein for more information on such
chemical reaction systems.

We transfer results from~\cite{Balding88} to Herman's protocol.  The
setup is slightly different because, unlike chemical particles, the
tokens in Herman's protocol move only in one direction.  This
difference is inconsequential, as the state of a system can be
captured using only relative token (or particle) distances.  Care must
be taken though, because Balding does not consider ``synchronous''
particle movement (this would make no sense in chemistry), but
particles moving ``asynchronously'' or continuously in a Brownian
motion.

Given two tokens $u$ and $v$ with $1 \le u < v \le M$, % and initial positions $z(u)$ and $z(v)$,
 we define a random variable~$\gT{uv}$ and events $\A{uv}{\down}$ and $\A{uv}{\up}$
 in terms of a system
 in which collisions between tokens $u$ and $v$ cause $u$ and~$v$ to be annihilated,
 but the movement of the other tokens and their possible collisions are ignored.
In that system, $\gT{uv}$ denotes the time until $u$ and $v$ have collided.
Further, let $\A{uv}{\down}$ and $\A{uv}{\up}$ denote the events that tokens $u$ and $v$ eventually collide {\em down} and {\em up}, respectively.
By colliding down (resp.\ up) we mean that, upon colliding, the token $u$ (resp.\ $v$) has caught up with~$v$ (resp.\ $u$) in clockwise direction;
 more formally, if $d_u, d_v \ge 0$ denote the distances travelled in clockwise direction by the tokens until collision,
 then the collision is said to be down (resp.\ up) if $z(u) + d_u = z(v) + d(v)$ (resp.\ $z(u) + d_u + N = z(v) + d_v$).
The behavior of two such tokens is equivalent to that of a one-dimensional random walk on $\{0, \ldots, N\}$, started at $z_{uv}$,
 with absorbing barriers at~$0$ and~$N$:
 the position in the random walk corresponds to the distance between the tokens,
 and colliding down (resp.\ up) corresponds to being absorbed at~$0$ (resp.\ $N$).
By this equivalence we have $\P{\A{uv}{\down}} = 1 - z_{uv} / N$ and $\P{\A{uv}{\up}} = z_{uv} / N$ (see, e.g., \cite{FellerVol168}).

Proposition~\ref{prop:Balding} below allows to express the distribution of~$\tim$
 in terms of the distribution of~$\gT{uv}$, conditioned under $\A{uv}{\down}$ and $\A{uv}{\up}$, respectively.
Those distributions are well-known \cite{FellerVol168,Cox01}.
For the statement we need to define the set $W_M$ of all {\em pairings}.
A pairing is a set $w = \{(u_1,v_1), \ldots, (u_m,v_m) \}$ with $1 \le u_i < v_i \le M$ for all~$i$, such that there is $w_0 \in \{1, \ldots, M\}$
 with $\{u_1, v_1, \ldots, u_m, v_m, w_0\} = \{1, \ldots, M\}$.
Define $s(w) = 1$ if the permutation $(u_1v_1 \cdots u_m v_m w_0)$ is even, and $s(w) = -1$ otherwise.
(This is well-defined: it is easy to see that $s(w)$ does not depend on the order of the $(u_i,v_i)$.)
We have the following proposition:

\newcommand{\stmtpropBalding}{
 Let $M \ge 3$.
 For all $t \ge 0$:
 \[
  \P{\tim \le t} = \sum_{w \in W_M} s(w) \prod_{(u,v) \in w}
   \left( \P{\gT{uv} \le t \;\cap\; \A{uv}{\down}} - \P{\gT{uv} \le t \;\cap\; \A{uv}{\up}} \right)\,.
 \]
}
\begin{proposition}[cf.\ {\cite[Theorem 2.1]{Balding88}}] \label{prop:Balding}
 \stmtpropBalding
\end{proposition}
Balding's Theorem 2.1 in~\cite{Balding88} is more general in that it gives a generating function
 for the number of remaining tokens at time~$t$.
Strictly speaking, Balding's theorem is not applicable to the synchronous version of Herman's protocol,
 because he only considers tokens that move according to the asynchronous version (in our terms), and tokens in a Brownian motion.
In addition, his proof omits many details, so we give a self-contained proof for Proposition~\ref{prop:Balding} in the appendix.

%We are primarily interested in the expectation of~$\tim$.
Theorem~\ref{thm:Balding} below yields an expression for~$\E \tim$.
We define the set $\dir{W_M}$ of all \emph{directed pairings} as the set of all sets
 $\dir{w} = \{(u_1, v_1, d_1), \ldots, (u_m,v_m,d_m)\}$ such that
  $\{(u_1,v_1), \ldots, (u_m,v_m)\} \in W_M$ and $d_i \in \{\down,\up\}$ for all $i \in \{1, \ldots, m\}$.
For a directed pairing $\dir{w} = \{(u_1, v_1, d_1), \ldots, (u_m, v_m, d_m)\}$ we define
 \[
  \dir{s}(\dir{w}) := s(\{(u_1, v_1), \ldots, (u_m,v_m)\}) \cdot (-1)^{|\{i \mid 1 \le i \le m, \ d_i = \up\}|}
 \]
and the event $\dA{w} := \bigcap_{i=1}^m \A{u_i v_i}{d_i}$.
Notice that $\P{\dA{w}} = \prod_{i=1}^m \P{\A{u_i v_i}{d_i}}$.
Further, we set $\dgT{w} := \max\{ \gT{u_i v_i} \mid 1 \le i \le m \}$.
We have the following theorem:
\newcommand{\stmtthmBalding}{
 For $M \ge 3$:
 \[
  \E \tim = \sum_{\dir{w} \in \dir{W_M}} \dir{s}(\dir{w}) \cdot \Ex{ \dgT{w} \mid \dA{w} } \cdot \P{\dA{w}}\,.
 \]
}
\begin{theorem} \label{thm:Balding}
 \stmtthmBalding
\end{theorem}

\paragraph{A Finite Expression for $\E \tim$.}
In the rest of the section we focus on the synchronous protocol.
We obtain a closed formula for~$\E \tim$ in Proposition~\ref{prop:expectation-finite-expression} below.

For $1 \le u < v < M$, we define $z_{u v \down} := z_{u v}$ and $z_{u v \up} := N - z_{u v}$.
For sets $\emptyset \ne \dir{x} \subseteq \dir{w} \in \dir{W_M}$ with
 $\dir{x} = \{(u_1,v_1,d_1), \ldots, (u_k,v_k,d_k)\}$ and
 $\dir{w} = \{(u_1,v_1,d_1), \ldots, (u_m,v_m,d_m)\}$
 we write
 \begin{align*}
  \yF(\dir{x}, \dir{w}) & := \left( \frac{z_{u_1 v_1 d_1}}{N}, \ldots, \frac{z_{u_k v_k d_k}}{N} \right)  \quad \text{and} \\
  \yG(\dir{x}, \dir{w}) & := \left( \frac{z_{u_{k+1} v_{k+1} d_{k+1}}}{N}, \ldots, \frac{z_{u_m v_m d_m}}{N} \right) \,.
 \end{align*}
Let
\[
 g(j,y;u) := \frac{\sin(j \pi y) \cdot \sin(j \pi u)}{1 - \cos(j \pi u)} \qquad \text{and} \qquad
 h(j;u)   := 1 - 2 r (1-r) \left( 1 - \cos (j \pi u) \right)\,,
\]
and define, for $k \in \N_+$ and $\ell \in \N_+$,
\begin{align*}
 \F{N}_k(y_1, \ldots, y_k) & := -\left(\frac{-1}{N}\right)^k \cdot \sum_{j \in \{1,\ldots,N-1\}^k}
      \frac{ \prod_{i=1}^k g( j(i), y_i; 1/N ) }{ 1 - \prod_{i=1}^k h(j(i); 1/N) } \qquad \text{and} \\
 \G_\ell(y_1, \ldots, y_\ell) & :=  \prod_{i=1}^\ell \left( 1 - y_i \right) \,.
\end{align*}
We drop the subscripts of $\F{N}_k$ and $\G_\ell$, if they are understood.
Observe that $\F{N}$ and $\G$ are continuous and do not depend on the order of their arguments.
The following proposition gives, for the synchronous protocol, a concrete expression for $\E \tim$.
\newcommand{\stmtpropexpectationfiniteexpression}{
 Consider the synchronous protocol.
 For $M \ge 3$:
 \[
  \E \tim = \sum_{\dir{w} \in \dir{W_M}} \dir{s}(\dir{w}) \sum_{\emptyset \ne \dir{x} \subseteq \dir{w}}
                 \F{N}(\yF(\dir{x}, \dir{w})) \cdot \G(\yG(\dir{x}, \dir{w}))\,.
 \]
}
\begin{proposition}\label{prop:expectation-finite-expression}
 \stmtpropexpectationfiniteexpression
\end{proposition}
%We have $|W_M| = M! / (m! \cdot 2^m)$.
%It follows that computing~$\E \tim$ according to Proposition~\ref{prop:expectation-finite-expression}
% costs $O(M! \cdot N^m)$ arithmetic operations (counting $\sin$ and~$\cos$ as basic operations).
%On the other hand,

\paragraph{An Approximation for $\E \tim$.}
The function~$\F{N}$ in Proposition~\ref{prop:expectation-finite-expression} depends on~$N$, and also on~$r$.
This prohibits a deeper analysis as needed in Section~\ref{sec:restabilization}.
Proposition~\ref{prop:expectation-continuous} gives an approximation of $\E \tim$ without those dependencies.
To state it, we define, for $k \in \N_+$, a function $\wF_k : [0,1]^k \to \R$ with
\begin{align*}
 \wF_k(y_1, \ldots, y_k) & = \frac{-1}{\pi^2} \left(\frac{-2}{\pi}\right)^k
   \sum_{j \in \N_+^k} \frac{ \prod_{i = 1}^k \sin(y_i j(i) \pi) }
                    { \left( \prod_{i = 1}^k j(i) \right) \left( \sum_{i=1}^k j(i)^2 \right) } \,.
\end{align*}
We drop the subscript of~$\wF_k$, if it is understood.
It follows from Lemma~\ref{lem:ben-sum-2} in the appendix that the series in~$\wF_k$ converges.
We have the following proposition.
\newcommand{\rint}{\left(\frac12 - \frac{\sqrt[4]{27}}{6}, \frac12 + \frac{\sqrt[4]{27}}{6}\right)}%
\newcommand{\ETapp}{\widetilde{E}}%
\newcommand{\stmtpropexpectationcontinuous}{
 Consider the synchronous protocol.
 Let
 \[
  \ETapp := \frac{N^2}{r (1-r)}
   \sum_{\dir{w} \in \dir{W_M}} \dir{s}(\dir{w})
   \sum_{\emptyset \ne \dir{x} \subseteq \dir{w}}
   \wF(\yF(\dir{x}, \dir{w})) \cdot \G(\yG(\dir{x}, \dir{w}))\,.
%    \frac{-1}{r (1-r) \pi^2} \left( \frac{-2}{\pi} \right)^{|\dir{x}|}
%      \left( \prod_{(u,v,d) \in \dir{w} \setminus \dir{x}} \left( 1 - z_{u v d}/N \right) \right)  \notag \\
%    \cdot \left( O(N^\varepsilon) + N^2 \cdot
%           \sum_{j \in \N_+^{\dir{x}}} \frac{ \prod_{(u,v,d) \in \dir{x}} \sin \frac{z_{u v d} j(u,v,d) \pi}{N} }
%                     { \left( \prod_{(u,v,d) \in \dir{x}} j(u,v,d) \right) \left( \sum_{(u,v,d) \in \dir{x}} j(u,v,d)^2 \right) }
%          \right) \;.
 \]
 Then, for each fixed $M \ge 3$ and $r \in \rint \approx \left(0.12, 0.88 \right)$ and $\varepsilon > 0$,
 \[
  \E \tim = \ETapp + O(N^\varepsilon)\,.
 \]
}
\begin{proposition} \label{prop:expectation-continuous}
 \stmtpropexpectationcontinuous
\end{proposition}
The proof of Proposition~\ref{prop:expectation-continuous} is elementary but involved.

\section{The Full Configuration} \label{sec:initial}

In this section we consider the initial configuration in which every
processor has a token, i.e., $N=M$.  We call this configuration
\emph{full}.  Notice that in the full configuration, with all bits set
to~$0$, in the successor configuration each bit is independently set
to $1$ with probability $r$.  Thus we study the full configuration in
lieu of the random configuration.
%
%Initializing with the full configuration is easily
%implemented by initializing all bits with~$0$.  One might think that
%initializing the bits with a random value may be better, because this
%immediately removes half of the tokens in expectation.  However, if
%the bits are initialized with~$0$, all bits are set randomly in the
%first step as well (assuming $r=\frac12$).  So the full configuration
%is a natural start.
%
We have the following theorem:
\newcommand{\stmtthmfull}{%
 For the synchronous protocol with parameter $r$ let $D = r(1-r)$.
 For the asynchronous protocol with parameter $\lambda > 0$ let $D =
 \lambda$.  For almost all odd $N \in \N_+$, we have for the full
 configuration:
 \begin{align*}
  \E \tim                 & \le 0.0285 N^2/D \quad \text{and} \quad \P{\tim  \ge 0.02 N^2/D} < 0.5 \,.
 \end{align*}
}
\begin{theorem} \label{thm:full}
 \stmtthmfull
\end{theorem}

Recall from Proposition~\ref{prop:triangle} that, for $N$ an odd
multiple of $3$, we have $\E \tim = \frac{1}{27} \frac{N^2}{D} \approx
0.0370 \frac{N^2}{D}$ if we start from the equilateral configuration.
It follows that, for large~$N$, the full configuration (with~$M=N$)
stabilizes faster than the equilateral configuration (with $M=3$).
This is consistent with the aforementioned conjecture of McIver and
Morgan that the equilateral configuration with $M=3$ is the worst case
among all configurations for a fixed $N$.

%in~\cite{McIverMorgan} that, for
%all~$N$, the equilateral configuration is the worst case, not only
%among the configurations with $M=3$, but among all configurations.  It
%also agrees with experiments carried out using the probabilistic model
%checker PRISM---see~\cite{PrismHerman}.

\section{Restabilization} \label{sec:restabilization}

In this section we restrict attention to the synchronous version of
Herman's algorithm and consider the standard bit-array implementation.
Theorem~\ref{thm:upper-bound} shows that the worst-case expected time
to termination, considering all initial configurations, is $\E \tim =
O(N^2)$.  We imagine that an initial configuration represents the
state of the system immediately after an error, that is, the ring of
tokens has become illegitimate because some of positions in the bit
array were corrupted.  In this light a natural restriction on initial
configurations is to consider those that arise from a one-token
configuration by corrupting some fixed number $m$ of bits.  We call
these \emph{flip-$m$} configurations.  Notice that, by the token
representation in Herman's protocol, a single bit error can lead to
the creation of two neighboring tokens.  So, $m$ bit errors could lead
to the creation of $m$ new pairs of neighboring tokens.  It could also
happen that two bit errors affect neighboring bits, leading to a new pair
of tokens at distance~$2$.  To account for this, we characterize
flip-$m$ configuration as those with at most $2 m + 1$ tokens such
that the tokens can be arranged into pairs, each pair at distance at
most~$m$, with one token left over.

%A crucial property of self-stabilizing systems is their
%``self-healing'' capability; i.e., even if errors happen, they return
%to a legitimate configuration.  Herman's protocol (we consider the
%synchronous version in this section) is self-healing, because any
%configuration leads to a legitimate one with probability~$1$.  By
%Theorem~\ref{thm:upper-bound} we have $\E \tim = O(N^2)$ for any
%initial configuration, no matter whether it is obtained by accidental
%bit errors or not.  However, we show that this expectation can be
%improved to $O(N)$, when starting in a \emph{flip-$m$} configuration,
%for any constant $m \ge 1$.  In particular, starting from a 1-token
%configuration, if $m$ bits flip erroneously, we obtain a flip-$m$
%configuration.  Notice that, by the token representation in Herman's
%protocol, a single bit error usually leads to the creation of two
%neighboring tokens.  So, $m$ bit errors could lead to $m$ pairs of
%neighboring tokens, plus the token from the beginning.  It could also
%happen that two bit errors cancel each other out or affect neighboring
%bits, leading to a pair of tokens with distance~$2$.  To account for
%this, we define a flip-$m$ configuration as any configuration with at
%most $2 m + 1$ tokens, where there is a token such that, when removed,
%the remaining tokens can be organized in pairs, each pair with
%distance at most~$m$.  

Fixing the number of bit errors we show that the expected time to
restabilization improves to $O(N)$.  Formally we show:
\begin{theorem} \label{thm:restabilization}
 Consider the synchronous protocol.
 Fix any $m \in \N_+$ and $r \in \rint \approx \left(0.12, 0.88 \right)$.
 Then for any flip-$m$ configuration we have $\E \tim = O(N)$.
\end{theorem}
\begin{proof}
It suffices to consider flip-$m$ configurations with $M = 2 m + 1$
tokens.  Without loss of generality, we assume that, when removing
token $2 m + 1$, the token pairs $(1,2), (3,4), \ldots, (2 m - 1, 2
m)$ have distances at most~$m$; i.e., we assume $z(u+1) - z(u) \le m$
for all odd $u$ between $1$ and $2 m - 1$.

\newcommand{\Cl}{\mathit{Cl}}
For each directed pairing $\dir{w} \in \dir{W_M}$, we define its \emph{class} $\Cl(\dir{w})$ and its \emph{companion pairing} $\dir{w}' \in \dir{W_M}$.
For the following definition, we define $\widetilde{u} := u+1$, if $u$ is odd, and $\widetilde{u} := u-1$, if $u$ is even.
\begin{itemize}
\item
 If $(u,M,d) \in \dir{w}$ for some $u$, then $\Cl(\dir{w}) = 0$.
 Its companion pairing is obtained, roughly speaking, by $u$ and $\widetilde{u}$ switching partners.
 More precisely:
 \begin{itemize}
 \item
  If $(\widetilde{u},v,d')$ (resp.\ $(v,\widetilde{u},d')$) for some $(v,d')$, then the companion pairing of~$w$ is obtained
   by replacing $(u,M,d)$ and $(\widetilde{u},v,d')$ with $(\widetilde{u},M,d)$ and $(u,v,d')$ (resp.\ $(v,u,d')$).
 \item
  Otherwise (i.e., $\widetilde{u}$ does not have a partner), the companion pairing of~$w$ is obtained
   by replacing $(u,M,d)$ with $(\widetilde{u},M,d)$.
 \end{itemize}
\item
 If $\dir{w} = \{(1,2,d_1), (3,4,d_2), \ldots, (M-2,M-1,d_m)\}$ for some $d_1, \ldots, d_m$,
  then $\Cl(\dir{w}) = m$.
 In this case, $\dir{w}$ does not have a companion pairing.
\item
 Otherwise, $\Cl(\dir{w})$ is the greatest number~$i$ such that for all $1 \le j \le i-1$,
  the tokens $2j-1$ and $2j$ are partners (i.e., $(2j-1, 2j, d)$ for some~$d$).
 Notice that $0 < \Cl(\dir{w}) < m$.
 The companion pairing of~$\dir{w}$ is obtained by $2i-1$ and $2i$ switching partners.
\end{itemize}
It is easy to see that, for any $\dir{w} \in \dir{W_M}$ with $\Cl(\dir{w}) < m$, we have $\Cl(\dir{w}) = \Cl(\dir{w}')$,
 and the companion pairing of~$\dir{w}'$ is~$\dir{w}$, and $\dir{s}(\dir{w}) = - \dir{s}(\dir{w}')$.
\newcommand{\dPlus}{\dir{W_M}^{(+)}}%
\newcommand{\dMinus}{\dir{W_M}^{(-)}}%
\newcommand{\dm}{\dir{W_M}^{(m)}}%
Partition $\dir{W_M}$ into the following sets:
 \begin{align*}
   \dPlus  & := \{ \dir{w} \in \dir{W_M} \mid \Cl(\dir{w}) < m \text{ and } \dir{s}(\dir{w}) = +1 \} && \text{and} \\
   \dMinus & := \{ \dir{w} \in \dir{W_M} \mid \Cl(\dir{w}) < m \text{ and } \dir{s}(\dir{w}) = -1 \} && \text{and} \\
   \dm     & := \{ \dir{w} \in \dir{W_M} \mid \Cl(\dir{w}) = m  \} \,.
 \end{align*}
The idea of this proof is that, in the sum of Proposition~\ref{prop:expectation-continuous},
 the terms from $\dPlus \cup \dMinus$ cancel each other ``almost'' out, and the terms from $\dm$ are small.
To simplify the notation in the rest of the proof,
 let $y(\dir{x}, \dir{w}) := (\yF(\dir{x},\dir{w}), \yG(\dir{x}, \dir{w}))$ and
 $H(y(\dir{x}, \dir{w})) := \wF(\yF(\dir{x}, \dir{w})) \cdot \G(\yG(\dir{x}, \dir{w}))$.
Since $\wF$ and $\G$ are continuous and bounded, so is~$H$.
\begin{itemize}
\item
 Let $(\dir{x}, \dir{w})$ with $\dir{x} \subseteq \dir{w} \in \dPlus \cup \dMinus$.
 To any such $(\dir{x}, \dir{w})$  we associate a companion $(\dir{x}', \dir{w}')$ such that $\dir{w}'$ is the companion pairing of~$\dir{w}$,
   and $\dir{x}' \subseteq \dir{w}'$ is obtained from~$\dir{x}$ in the following way:
    if $\dir{w}'$ is obtained from~$\dir{w}$ by replacing one or two triples $(u,v,d)$,
     then $\dir{x}'$ is obtained by performing the same replacements on~$\dir{x}$ (of course, only if $(u,v,d) \in \dir{x}$).
 Note that $y(\dir{x}, \dir{w})$ and $y(\dir{x}',\dir{w}')$ are equal in all components,
  except for one or two components, where they differ by at most~$\frac{m}{N}$.
 Hence we have (for constant~$m$) that
  \begin{equation*}
    y(\dir{x}', \dir{w}') = y(\dir{x}, \dir{w}) + O(1/N) \cdot (1, \ldots, 1)\,.
  \end{equation*}
 Since $H$ is continuous, it follows
  \[
   H(y(\dir{x}', \dir{w}')) = H(y(\dir{x}, \dir{w})) + O(1/N) \,.
  \]
\item
 Let $(\dir{x}, \dir{w})$ with $\dir{x} \subseteq \dir{w} \in \dm$.
 Note that all components of $\yF(\dir{x}, \dir{w})$ are at most~$\frac{m}{N}$ or at least $1 - \frac{m}{N}$.
 Also note that for any vector $e \in \{0,1\}^{|\dir{x}|}$ it holds $H(e,\yG(\dir{x},\dir{w})) = 0$.
 Since $H$ is continuous, it follows
  \[
   H(y(\dir{x}, \dir{w})) = O(1/N) \,.
  \]
\end{itemize}
Take $0 < \varepsilon < 1$.
By Proposition~\ref{prop:expectation-continuous} and the above considerations, we have:
\begin{align*}
 \E \tim & = O(N^\varepsilon) +
   \frac{N^2}{r (1-r)}
   \sum_{\dir{w} \in \dir{W_M}} \dir{s}(\dir{w})
   \sum_{\emptyset \ne \dir{x} \subseteq \dir{w}}
   H(y(\dir{x}, \dir{w})) \\
         & = O(N^\varepsilon) +
   \frac{N^2}{r (1-r)} \cdot \left(
   \sum_{\dir{w} \in \dPlus}   \quad \sum_{\emptyset \ne \dir{x} \subseteq \dir{w}} H(y(\dir{x}, \dir{w})) \right. \\
         & \hspace{50mm}     - \sum_{\emptyset \ne \dir{x}' \subseteq \dir{w}'} H(y(\dir{x}', \dir{w}')) \\
         & \hspace{33mm} \left. + \sum_{\dir{w} \in \dm} \quad \sum_{\emptyset \ne \dir{x} \subseteq \dir{w}} H(y(\dir{x}, \dir{w})) \right) \\
         & = O(N^\varepsilon) +
   \frac{N^2}{r (1-r)} \cdot \left(
   \sum_{\dir{w} \in \dPlus}   \quad \sum_{\emptyset \ne \dir{x} \subseteq \dir{w}} O(1/N) \right. \\
         & \hspace{33mm} \left. + \sum_{\dir{w} \in \dm} \quad \sum_{\emptyset \ne \dir{x} \subseteq \dir{w}} O(1/N) \right) \\
         & = O(N^\varepsilon) + O(N) = O(N) \,.
\end{align*}
\qed
\end{proof}

\section{Conclusions and Future Work} \label{sec:conclusions}

We have obtained several results on the expected self-stabilization
time $\E \tim$ in Herman's algorithm.  We have improved the best-known
upper bound for arbitrary configurations, and we have given new and
significantly better bounds for special classes of configurations: the
full configuration, the random configuration, and, in particular, for
configurations that arise from a fixed number of bit errors.  For the
latter class, $\E \tim$ reduces to $O(N)$, pointing to a previously
unknown feature that Herman's algorithm recovers quickly from bounded
errors.  We have also shown that an asynchronous version of Herman's
algorithm not requiring synchronization behaves similarly.  For our
analysis, we have transferred techniques that were designed for the
analysis of chemical reactions.  

The conjecture of~\cite{McIverMorgan}, saying that the equilateral
configuration with three tokens constitutes the worst-case, remains
open.  We hope to exploit our closed-form expression for $\E \tim$ to
resolve this intriguing problem.  While we have already shown that
many relevant initial configurations provably converge faster, solving
this conjecture would close the gap between the lower and upper bounds
for stabilization time for arbitrary configurations.  We would also
like to investigate the performance of the algorithm in case the number
of bit errors is not fixed, but is small (e.g., logarithmic) in the
number of processes.

%---------------------
\bibliographystyle{plain} %oder alpha oder splncs
\bibliography{db}
%---------------------

\iftechrep{
\newpage
\appendix
\section{Proof of Theorem~\ref{thm:upper-bound}}

Here is a restatement of Theorem~\ref{thm:upper-bound}:
\begin{qtheorem}{\ref{thm:upper-bound}}
 \stmtthmupperbound
\end{qtheorem}
\begin{proof}
 We build upon the proof in~\cite{Nakata05} for the synchronous case, which works as follows.
 For $M \ge 3$, let $\tau_M$ denote the maximal expected time for a configuration with $M$ tokens to reach a configuration with fewer than~$M$ tokens,
  where the maximum is taken over all $M$-token configurations.
 It is shown that $\tau_M \le \frac{1}{M^2} \cdot \frac{N^2}{D}$.
 Since $\E \tim \le \tau_3 + \tau_5 + \tau_7 + \cdots$ and $\frac{1}{1^2} + \frac{1}{3^2} + \frac{1}{5^2} + \cdots = \frac{\pi^2}{8}$,
  it follows that $\E \tim \le \left( \frac{\pi^2}{8} - 1 \right) \cdot \frac{N^2}{D}$.
 We obtain the improvement by replacing the bound $\tau_3 \le \frac{1}{9} \cdot \frac{N^2}{D}$ with $\tau_3 \le \frac{1}{27} \cdot \frac{N^2}{D}$,
  which follows from Proposition~\ref{prop:triangle} and the comments below the proposition.

 To generalize the result to the asynchronous case,
  one needs to show that $\tau_M \le \frac{1}{M^2} \cdot \frac{N^2}{D}$ also holds in the asynchronous case.
 Before showing how to suitably adapt the proof in~\cite{Nakata05},
  we first provide more details on the proof in~\cite{Nakata05} for the synchronous case.
 Let $M \ge 3$.
 For a configuration~$c$ with at most $M$ tokens, define $\delta_M(c)$ as follows:
  if $c$ has less than~$M$ tokens, then $\delta_M(c) = 0$; otherwise $\delta_M(c)$ is the minimal token distance in~$c$.
 Let $c'$ the successor configuration of~$c$.
 Note that $\delta_M(c)$ and $\delta_M(c')$ differ by at most~$1$.
 Also note that a given token pair decreases its distance by~$1$ with probability $r(1-r)$, because one token must be passed, the other one kept.
 Similarly, the distance is increased by~$1$ also with probability $r(1-r)$.
 For the event that $\delta_M$ decreases by~$1$, it suffices that the distance decreases for \emph{one} token pair
  among those that define the minimal distance~$\delta_M(c)$.
 For the event that $\delta_M$ increases by~$1$, the distance must increase for \emph{all} token pairs which define~$\delta_M(c)$.
 It follows:
 \begin{align*}
  \P{\delta_M(c') = \delta_M(c) - 1 \mid c \text{ and } 1 \le \delta_M(c) \le \lfloor N/M \rfloor} \quad & \ge \quad r (1-r) \\
  \P{\delta_M(c') = \delta_M(c) + 1 \mid c \text{ and } 1 \le \delta_M(c) \le \lfloor N/M \rfloor - 1} \quad & \le \quad r (1-r) \,.
%  \P{\delta_M(c') = \delta_M(c) + 1 \mid c \text{ and }  \delta_M(c) = \lfloor N/M \rfloor} \quad & = \quad 0 \,.
 \end{align*}
 This process is compared in~\cite{Nakata05} with the following random walk on $\{0, \ldots, \lfloor N/M \rfloor\}$, absorbing at state~$0$:
 \begin{align*}
  \P{X' = X - 1 \mid 1 \le X \le \lfloor N/M \rfloor} \quad & = \quad r (1-r) \\
  \P{X' = X + 1 \mid 1 \le X \le \lfloor N/M \rfloor - 1} \quad & = \quad r (1-r) \\
  \P{X' = X  \mid 1 \le X \le \lfloor N/M \rfloor - 1} \quad & = \quad 1 - 2 r (1-r) \\
  \P{X' = X  \mid X = \lfloor N/M \rfloor} \quad & = \quad 1 - r (1-r) \,.
 \end{align*}
 It is argued there that the expected time to hit~$0$ in this random walk is an upper bound on the expected time to hit a configuration~$c$
  with $\delta_M(c) = 0$, and hence also on $\tau_M$.
 The expected time to hit~$0$ in the random walk is maximized when starting at $X  = \lfloor N/M \rfloor$,
  in which case the expected time is $\frac{\lfloor N/M \rfloor ( \lfloor N/M \rfloor + 1 )}{2 r (1-r)} \le \frac{1}{M^2} \cdot \frac{N^2}{D}$.

 This argument can be adapted to the asynchronous protocol in a straightforward way:
 Arguing similarly as above, the rate in which $\delta_M$ decreases by~$1$ is \emph{at least}~$\lambda$,
  and the rate in which $\delta_M$ increases by~$1$ is \emph{at most}~$\lambda$.
 We compare this process with a continuous-time Markov chain on $\{0, \ldots, \lfloor N/M \rfloor\}$, absorbing at state~$0$:
 \begin{itemize}
  \item
   the rate in which $X$ is decreased by~$1$ is $\lambda$;
  \item
   the rate in which $X$ is increased by~$1$ is $\lambda$ if $1 \le X \le \lfloor N/M \rfloor - 1$; and $0$ if $X = \lfloor N/M \rfloor$.
 \end{itemize}
 Analogous arguments yield
  \[
   \tau_M \le \frac{\lfloor N/M \rfloor ( \lfloor N/M \rfloor + 1 )}{2 \lambda } \le \frac{1}{M^2} \cdot \frac{N^2}{D}\,.
  \]
\qed
\end{proof}

\section{Proof of Proposition~\ref{prop:Balding}}

The proof follows the one of Theorem~2.1 of~\cite{Balding88}, but is more detailed and applies also to the synchronous version of Herman's protocol.

We first prove the following lemma.
\begin{lemma} \label{lem:auxBalding}
\newcommand{\gP}[3]{P_{(#1)#2,#3}}
 Let $M \ge 3$.
 Denote, for $1 \le u < v \le M$, by $\tim_{-uv}$ the time until one token is left,
  in a system with $M-2$ tokens obtained by removing the $u$-th and the $v$-th token.
 Then, for all $t \ge 0$:
 \[
  \P{\tim \le t} = \frac{1}{m} \sum_{1 \le u < v \le M} (-1)^{v - u - 1}
   \left( \P{\gT{uv} \le t \;\cap\; \A{uv}{\down}} - \P{\gT{uv} \le t \;\cap\; \A{uv}{\up}} \right) \P{\tim_{-uv} \le t}\,.
 \]
\end{lemma}
\begin{proof}%
\newcommand{\De}[3]{D_{(#1)#2}^{#3}}
 Consider, for $1 \le u < v \le M$, the expression
 \begin{align}
    & \left( \P{\gT{uv} \le t \;\cap\; \A{uv}{\down}} - \P{\gT{uv} \le t \;\cap\; \A{uv}{\up}} \right) \P{\tim_{-uv} \le t} \notag \\
  = & \P{\gT{uv} \le t \;\cap\; \A{uv}{\down}} \P{\tim_{-uv} \le t} - \P{\gT{uv} \le t \;\cap\; \A{uv}{\up}} \P{\tim_{-uv} \le t}\,. \label{eq:auxBaldingDiff}
 \end{align}
 We wish to define events $\De{uv}{\down}{}$ and $\De{uv}{\up}{}$ such that
 \begin{align*}
  \P{\De{uv}{\down}{}} & = \P{\gT{uv} \le t \;\cap\; \A{uv}{\down}} \P{\tim_{-uv} \le t}
   \intertext{and}
  \P{\De{uv}{\up}{}}   & = \P{\gT{uv} \le t \;\cap\; \A{uv}{\up}} \P{\tim_{-uv} \le t}\,.
 \end{align*}
 This can be done as follows.
 Call the tokens $u$ and $v$ {\em red}, and the other tokens {\em green}.
 Think of a system in which red and green tokens do not interact, i.e., red-green meetings do not cause annihilations.
 Meeting tokens of the same color are however annihilated.
 Then $\De{uv}{\down}{}$ can be defined as the event that, by time $t$, the red tokens $u$ and $v$ have met down,
  and all other tokens have annihilated, except for one remaining (green) token.
 The event $\De{uv}{\up}{}$ is defined similarly.
 With this definition, the expression in~\eqref{eq:auxBaldingDiff} is equal to
  $\P{\De{uv}{\down}{}} - \P{\De{uv}{\up}{}}$.

 Now we partition the event $\De{uv}{\down}{}$ according to the first red-green meeting as follows:
 \[
  \De{uv}{\down}{} = \De{uv}{\down}{0} \cup \mathop{\bigcup_{p \in \{u,v\}}}_{q \in \{1, \ldots, M\} \setminus \{u,v\}} \De{uv}{\down}{pq}\,,
 \]
 where the unions are disjoint, $\De{uv}{\down}{pq}$ is the event that the first red-green meeting is between $p$ and~$q$,
  and $\De{uv}{\down}{0}$ is the event that no red-green meeting occurs.
 If it happens that $u$ and $v$ have their first meeting with a green token (say, with $g_u$ and $g_v$, respectively) {\em at the same time},
  then we count this sample run in~$\De{uv}{\down}{ug_u}$.
 The event $\De{uv}{\up}{}$ is partitioned similarly;
  in particular, if $u$ and $v$ have their first meeting with a green token (say, with $g_u$ and $g_v$, respectively) at the same time,
  then we count this sample run in~$\De{uv}{\up}{vg_v}$.

 We show that each nonempty event $\De{uv}{\down}{pq}$ has a ``companion'' event with the same probability.
 \begin{itemize}
  \item
   Consider $\De{uv}{\down}{ug}$ with $g < v$.
   Its companion event is $\De{gv}{\down}{gu}$.
   In order to prove that those events have the same probability,
    we establish a bijection between $\De{uv}{\down}{ug}$ and~$\De{gv}{\down}{gu}$.
   The bijection~$b$ is defined as follows:
   Let $\omega$ be a sample run (up to time~$t$) of $\De{uv}{\down}{ug}$.
   Let $t_0 \le t$ be the time of the first red-green meeting in~$\omega$, i.e., $u$ and~$g$ meet at~$t_0$.
   Then $b(\omega)$ equals $\omega$, except that after time~$t_0$,
        the movement of token~$u$ in~$b(\omega)$ is the movement of token~$g$ in~$\omega$,
    and the movement of token~$g$ in~$b(\omega)$ is the movement of token~$u$ in~$\omega$.
   By the reflection principle, $\omega$ and $b(\omega)$ have the same probability.
   Furthermore, it is straightforward to verify that any sample run~$\omega$ is in~$\De{uv}{\down}{ug}$
    if and only if $b(\omega) \in \De{gv}{\down}{gu}$.

   Note that $\De{uv}{\down}{ug}$ and $\De{gv}{\down}{gu}$ are nonempty only if $u-g$ is odd,
    because all tokens between $u$ and $g$ must annihilate,
    so their number must be even.
%   To see, why these events have the same probability, consider a single sample run of $\De{uv}{\down}{ug}$.
%   By definition, there is a time~$t_0$ with $0 < t_0 < t$ such that at time~$t_0$ the tokens $u$ and~$g$ meet,
%    and this is the first red-green meeting.
%   In particular, the token~$v$ does not meet any token before time~$t_0$, and any green tokens between $u$ and~$g$ annihilate each other before~$t_0$.
%   After $t_0$, the remaining green tokens annihilate (except for the last one), and $u$ and~$v$ meet down.
  \item
   Similarly, for $g > v$, the companion event of $\De{uv}{\down}{ug}$ is $\De{vg}{\up}{gu}$.
   The events are nonempty only if $u-g$ is even.
  \item
   For $g < u$, the companion event of $\De{uv}{\down}{vg}$ is $\De{gu}{\up}{gv}$.
   The events are nonempty only if $v-g$ is even.
  \item
   For $g > u$, the companion event of $\De{uv}{\down}{vg}$ is $\De{ug}{\down}{gv}$.
   The events are nonempty only if $v-g$ is odd.
 \end{itemize}
 Similarly, there is a companion event to each nonempty event $\De{uv}{\up}{pq}$.
 Letting $\textup{RHS}$ denote the right hand side of the equation in the statement of the lemma, we have:
 \begin{align*}
  \textup{RHS}
  & = \frac{1}{m} \sum_{1 \le u < v \le M} (-1)^{v - u - 1}
       \left( \P{\gT{uv} \le t \;\cap\; \A{uv}{\down}} - \P{\gT{uv} \le t \;\cap\; \A{uv}{\up}} \right) \P{\tim_{-uv} \le t} \\
  & = \frac{1}{m} \sum_{1 \le u < v \le M} (-1)^{v - u - 1}
       \left( \P{\De{uv}{\down}{}} - \P{\De{uv}{\up}{}} \right) \\
  & = \frac{1}{m} \sum_{1 \le u < v \le M} (-1)^{v - u - 1}
       \left( \P{\De{uv}{\down}{0}} - \P{\De{uv}{\up}{0}} \right) \,,
 \end{align*}
   where the last equality is because the probabilities of the events $\De{uv}{\down}{pq}$ and~$\De{uv}{\up}{pq}$
    cancel with the probabilities of their respective companion events.
   The event $\De{uv}{\down}{0}$ is nonempty if and only if $v-u$ is odd;
   similarly, $\De{uv}{\up}{0}$ is nonempty if and only if $v-u$ is even.
   Hence, we have
 \begin{align*}
  \textup{RHS} & = \frac{1}{m} \sum_{1 \le u < v \le M} \P{\De{uv}{}{0}} \,,
 \end{align*}
 where $\De{uv}{}{0} := \De{uv}{\down}{0} \cup \De{uv}{\up}{0}$.
 Note that $\De{uv}{}{0}$ contains exactly those sample runs in which, under the normal annihilation rules,
  by time~$t$, the tokens $u$ and $v$ have met and annihilated, and all other tokens except for one have also annihilated.

 Recall that $W_M$ is the set of pairings.
 For any pairing $w = \{(u_1,v_1), \ldots, (u_m,v_m)\} \in W_M$ we denote by $E_w$ the event that, by time~$t$,
  for all $i \in \{1, \ldots, m\}$, the tokens $u_i$ and $v_i$ have met and annihilated (under the normal annihilation rules).
 Note that
  \[
   \De{uv}{}{0} = \bigcup_{w: (u,v) \in w \in W_M} E_w\,,
  \]
 where the union is disjoint.
 Hence, we have:
  \begin{align*}
   \textup{RHS}
   & = \frac{1}{m} \sum_{1 \le u < v \le M} \P{\De{uv}{}{0}} \\
   & = \frac{1}{m} \sum_{1 \le u < v \le M} \quad \sum_{w: (u,v) \in w \in W_M} \P{E_w} \\
   & = \sum_{w \in W_M} \P{E_w} \\
   & = \P{\bigcup_{w \in W_M} E_w} \\
   & = \P{\tim \le t} \,,
  \end{align*}
 which concludes the proof of the lemma.
\qed
\end{proof}

Now we can prove Proposition~\ref{prop:Balding} which is restated here.

\begin{qproposition}{\ref{prop:Balding}}
 \stmtpropBalding
\end{qproposition}
\begin{proof}
 The proof is by induction on $M = 3, 5, 7, \ldots$.
 The case $M=3$ is immediate from Lemma~\ref{lem:auxBalding}.
 (Notice in particular that $\P{\tim_{-uv} \le t} = 1$ if $M=3$.)

 For the induction step, let $M \ge 5$.
 By Lemma~\ref{lem:auxBalding} we have
 \begin{align*}
  \P{\tim \le t}
  & = \frac{1}{m} \sum_{1 \le u < v \le M} (-1)^{v - u - 1}
       \left( \P{\gT{uv} \le t \;\cap\; \A{uv}{\down}} - \P{\gT{uv} \le t \;\cap\; \A{uv}{\up}} \right) \P{\tim_{-uv} \le t}\,.
  \intertext{%
   For $1 \le u < v \le M$, we define the set $W_{-uv}$ similarly to the set $W_M$,
    but $W_{-uv}$ is the set of pairings on $\{1, \ldots, M\} \setminus \{u,v\}$ rather than on $\{1, \ldots, M\}$.
   Similarly, for $w' \in W_{-uv}$, the number $s'(w') \in \{-1,+1\}$ is defined as $s(w)$,
    but depending on the parity of the permutation of $\{1, \ldots, M\} \setminus \{u,v\}$.
   Applying the induction hypothesis we have:
  }
  \P{\tim \le t}
  & = \frac{1}{m} \sum_{1 \le u < v \le M} (-1)^{v - u - 1}
       \left( \P{\gT{uv} \le t \;\cap\; \A{uv}{\down}} - \P{\gT{uv} \le t \;\cap\; \A{uv}{\up}} \right) \\
  & \qquad \mbox{} \cdot \sum_{w' \in W_{-uv}} s'(w') \prod_{(u',v') \in w'}
      \left( \P{\gT{u' v'} \le t \;\cap\; \A{u' v'}{\down}} - \P{\gT{u' v'} \le t \;\cap\; \A{u' v'}{\up}} \right) \,.
  \intertext{%
   We claim that for any $w' \in W_{-u v}$, we have $(-1)^{v-u-1} s'(w') = s(w \cup \{(u,v)\})$.
   To see this, assume $w' = \{(u_1,v_1), \ldots, (u_{m-1},v_{m-1})\}$ and
    $\{u_1,v_1, \ldots, u_{m-1},v_{m-1}, w_0\} = \{1, \ldots, M\} \setminus \{u,v\}$.
   We need to argue that the parities of the permutations
        $p_1 = (u_1 v_1 \cdots u_{m-1} v_{m-1} w_0)$ and
    $p_2 = (u v u_1 v_1 \cdots u_{m-1} v_{m-1} w_0)$ are equal if and only if $v-u$ is odd.
   It suffices to argue that adding $u,v$ at the front of~$p_1$ adds an even number of inversions in the permutation,
    if and only if $v-u$ is odd.
   Since $u<v$, the pair $(u,v)$ is not an inversion.
   For $x < u$, both $(u,x)$ and $(v,x)$ are inversions.
   For $x > v$, neither $(u,x)$ nor $(v,x)$ are inversions.
   For $x \in \{u+1, \ldots, v-1\}$, the pair $(u,x)$ is not an inversion, but $(v,x)$ is.
   There is an even number of such~$x$, if and only if $v-u$ is odd.
   This proves the claim.
   It follows:
  }
  \P{\tim \le t}
  & = \frac{1}{m} \sum_{1 \le u < v \le M} \quad \sum_{w: (u,v) \in w \in W_M} s(w) \\
  & \qquad \mbox{} \cdot
   \prod_{(u,v) \in w} \left( \P{\gT{u v} \le t \;\cap\; \A{u v}{\down}} - \P{\gT{u v} \le t \;\cap\; \A{u v}{\up}} \right) \\
  & = \sum_{w \in W_M} s(w)
   \prod_{(u,v) \in w} \left( \P{\gT{uv} \le t \;\cap\; \A{uv}{\down}} - \P{\gT{uv} \le t \;\cap\; \A{uv}{\up}} \right)\,,
 \end{align*}
 which completes the induction proof.
\qed
\end{proof}

\section{Proof of Theorem~\ref{thm:Balding}}

Theorem~\ref{thm:Balding} is restated here:

\begin{qtheorem}{\ref{thm:Balding}}
 \stmtpropBalding
\end{qtheorem}

\begin{proof}
% For a directed pairing $\dir{w} \in \dir{W_M}$ we denote by $w$ its associated pairing in~$W_M$.
 By Proposition~\ref{prop:Balding} we have:
 \begin{align}
  \P{\tim > t}
   & = 1 - \P{\tim \le t} \notag \\
   & = 1 - \sum_{w \in W_M} s(w) \prod_{(u,v) \in w}
        \left( \P{\gT{uv} \le t \;\cap\; \A{uv}{\down}} - \P{\gT{uv} \le t \;\cap\; \A{uv}{\up}} \right) \notag \\
   & = 1 - \sum_{\dir{w} \in \dir{W_M}} \dir{s}(\dir{w}) \prod_{(u,v,d) \in \dir{w}} \P{\gT{uv} \le t \;\cap\; \A{uv}{d}} \notag \\
   & = 1 - \sum_{\dir{w} \in \dir{W_M}} \dir{s}(\dir{w}) \cdot \P{\dgT{w} \le t \;\cap\; \dA{w}} \,. \label{eq:thmBalding1}
 \end{align}
 The Markov chain associated with Herman's protocol has a unique bottom SCC.
 Hence, $\P{\tim = \infty} = 0$.
 Similarly, $\P{\dgT{w} = \infty} = 0$ for all $\dir{w} \in \dir{W_M}$.
 By~\eqref{eq:thmBalding1} it follows
  \[
   1 = \sum_{\dir{w} \in \dir{W_M}} \dir{s}(\dir{w}) \cdot \P{\dA{w}} \,,
  \]
 and hence
 \begin{align}
  \P{\tim > t}
  & = \sum_{\dir{w} \in \dir{W_M}} \dir{s}(\dir{w}) \cdot \P{\dgT{w} > t \;\cap\; \dA{w}} \notag \\
  & = \sum_{\dir{w} \in \dir{W_M}} \dir{s}(\dir{w}) \cdot \P{\dgT{w} > t \mid \dA{w}} \cdot \P{\dA{w}} \,. \label{eq:thmBaldingGreater}
 \end{align}
 For any random variable $X$ on~$\{0, 1, \ldots \}$, it is known that $\E X = \sum_{t=0}^\infty \P{X > t}$.
 Similarly, if $X$ is on~$[0, \infty)$, then $\E X = \int_{t=0}^\infty \P{X > t}\, dt$, see \cite{FellerVol168}.
 Hence, summing or integrating~\eqref{eq:thmBaldingGreater} over~$t$ yields the result.
\qed
\end{proof}

\section{Proof of Proposition~\ref{prop:expectation-finite-expression}}

Proposition~\ref{prop:expectation-finite-expression} is restated here.

\begin{qproposition}{\ref{prop:expectation-finite-expression}}
 \stmtpropexpectationfiniteexpression
\end{qproposition}
\begin{proof}
Given $\dir{w} \in \dir{W_M}$ and any $\dir{x} \subseteq \dir{w}$, we define $\dA{x} := \bigcap_{(u,v,d) \in \dir{x}} \A{u v}{d}$.
Recall that
 \begin{equation}
  \P{A_{\dir{w} \setminus \dir{x}}} = \prod_{(u,v,d) \in \dir{w} \setminus \dir{x}} \P{\A{u v}{d}}
    = \prod_{(u,v,d) \in \dir{w} \setminus \dir{x}} ( 1 - z_{u v d}/N ) = G(\yG(\dir{x},\dir{w}))\;.
   \label{eq:directed-independence}
 \end{equation}
The maximum-minimums identity states, for any set~$S$ of numbers, that $\max S = \sum_{\emptyset \ne S' \subseteq S} (-1)^{|S'|+1} \min S'$.
Using Theorem~\ref{thm:Balding} and the maximum-minimums identity, we get
 \begin{align*}
  \E \tim & = \sum_{\dir{w} \in \dir{W_M}} \dir{s}(\dir{w}) \Ex{ \dgT{w} \mid \dA{w} } \P{\dA{w}} \\
          & = \sum_{\dir{w} \in \dir{W_M}} \dir{s}(\dir{w}) \Ex{ \max \{ \gT{u v} \mid (u,v,d) \in \dir{w} \} \mid \dA{w} } \P{\dA{w}} \\
          & = \sum_{\dir{w} \in \dir{W_M}} \dir{s}(\dir{w}) \sum_{\emptyset \ne \dir{x} \subseteq \dir{w}} (-1)^{|\dir{x}|+1} \cdot
                \Ex{ \min \{ \gT{u v} \mid (u,v,d) \in \dir{x} \} \mid \dA{w} } \P{\dA{w}}  \\
          & = \sum_{\dir{w} \in \dir{W_M}} \dir{s}(\dir{w}) \sum_{\emptyset \ne \dir{x} \subseteq \dir{w}}
                   -(-1)^{|\dir{x}|} \cdot
                \Ex{ \min \{ \gT{u v} \mid (u,v,d) \in \dir{x} \} \mid \dA{x} } \P{\dA{x}} \P{A_{\dir{w} \setminus \dir{x}}}\;.
 \end{align*}
Consequently, by~\eqref{eq:directed-independence} it suffices to show
 \begin{equation}
  \Ex{ \min \{ \gT{u v} \mid (u,v,d) \in \dir{x} \} \mid \dA{x} } \P{\dA{x}}
   = \frac{1}{N^{|\dir{x}|}} \sum_{j \in \{1,\ldots,N-1\}^{\dir{x}}}
     \frac{ \prod_{(u,v,d) \in \dir{x}} g\left( j(u,v,d), \frac{z_{u v d}}{N}; \frac{1}{N} \right) }
          { 1 - \prod_{(u,v,d) \in \dir{x}} h\left(j(u,v,d); \frac{1}{N}\right) }   \,.
   \label{eq:proof-finite-expr-remains}
 \end{equation}
For any $(u,v,d)$, it follows from \cite{Cox01} (Section 2.2, Equation~(25)) that
 \[
  \P{\gT{u v} > t \;\cap\; \A{u v}{d}} = \frac{1}{N} \sum_{j=1}^{N-1} g\left(j, \frac{z_{u v d}}{N}; \frac{1}{N}\right) h\left(j; \frac{1}{N} \right)^t\,.
 \]
For any $\dir{x}$, we therefore have
 \begin{equation} \label{eq:proof-finite-expr-end}
 \begin{split}
  &\P{ \min \{ \gT{u v} \mid (u,v,d) \in \dir{x} \} > t \;\cap\; \dA{x} }\\
  & = \frac{1}{N^{|\dir{x}|}}  \sum_{j \in \{1,\ldots,N-1\}^{\dir{x}}}
       \left( \prod_{(u,v,d) \in \dir{x}} g\left( j(u,v,d), \frac{z_{u v d}}{N}; \frac{1}{N} \right) \right)
       \left( \prod_{(u,v,d) \in \dir{x}} h\left( j(u,v,d); \frac{1}{N} \right) \right)^t \,.
 \end{split}
 \end{equation}
Summing~\eqref{eq:proof-finite-expr-end} over $t=0, 1, \ldots$ yields~\eqref{eq:proof-finite-expr-remains}.
\qed
\end{proof}

\section{Proof of Proposition~\ref{prop:expectation-continuous}} \label{app:prop:expectation-continuous}

\newcommand{\stmtlemfkconvfinite}{
 For any fixed $k \in \N_+$ and $r \in \rint$, we have
 \begin{equation*}
  \sum_{\vj \in \{1, \ldots, N-1\}^k} \left| \of_k(\vj, \vy; 1/N) -  a_0(\vj) \right| = O \left(\frac{(\log N)^k}{N^2}\right) \,.
 \end{equation*}
}%
\newcommand{\stmtlemfkconvinfinite}{
 For any fixed $k \in \N_+$ and $\varepsilon > 0$, we have
  \[
   \sum_{j_k = N}^\infty \sum_{(j_1, \ldots, j_{k-1}) \in \N_+^{k-1}} \frac{1}{j_1 \cdot \ldots \cdot j_k \cdot (j_1^2 + \cdots + j_k^2)}
     = O \left( \frac{1}{N^{2-\varepsilon}} \right) \,.
  \]
}%

In this section we prove Proposition~\ref{prop:expectation-continuous}, which is restated here:

\begin{qproposition}{\ref{prop:expectation-continuous}}
 \stmtpropexpectationcontinuous
\end{qproposition}
\begin{proof}
In the following we write $\vj = (j_1, \ldots, j_k)$ and $\vy = (y_1, \ldots, y_k)$ for elements of $\N_+^k$ and $[0,1]^k$, respectively, where $k \in \N_+$.
Define the function
\begin{align*}
 f_k(\vj, \vy; u) & := \frac{\prod_{i=1}^k g(j_i,y_i;u) \cdot u^{k+2}}{1 - \prod_{i=1}^k h(j_i;u)} \,.
\end{align*}
Proposition~\ref{prop:expectation-finite-expression} then reads as
 \begin{multline}
  \E \tim = N^2 \cdot \sum_{\dir{w} \in \dir{W_M}} \dir{s}(\dir{w})
              \sum_{\emptyset \ne \dir{x} \subseteq \dir{w}}   -(-1)^{|\dir{x}|} \cdot \G(\yG(\dir{x}, \dir{w})) \notag \\
  \mbox{} \cdot \sum_{\vj \in \{1, \ldots, N-1\}^{|\dir{x}|}} f_k(\vj, \yF(\dir{x}, \dir{w}); 1/N) \,.
 \end{multline}
Consequently, it suffices to show that, for any fixed $k \in \N_+$ and $r \in \rint$ and $\varepsilon > 0$,
 \begin{equation}
  \sum_{\vj \in \{1, \ldots, N-1\}^k} f_k(\vj, \vy; 1/N)
     = \frac{\wF(\vy)}{ -(-1)^k \cdot r \cdot (1-r)} + O \left( \frac{1}{N^{2-\varepsilon}} \right) \;.
 \label{eq:prop-exp-cont-to-show-1}
 \end{equation}
Let
 \begin{align*}
  a_0(\vj)    & := \frac{2^k}{r \cdot (1-r) \cdot \pi^{k+2} \cdot j_1 \cdots j_k \cdot (j_1^2 + \cdots + j_k^2)} & \text{and} \\
  s(\vj, \vy) & := \sin(y_1 j_1 \pi) \cdots \sin(y_k j_k \pi) & \text{and} \\
 \of_k(\vj; u)      & := \frac{f_k(\vj, \vy; u)}{s(\vj, \vy)} \,.
 \end{align*}
Note that $\of_k(\vj; u)$ is independent of~$\vy$.
Then \eqref{eq:prop-exp-cont-to-show-1} is equivalent to
\begin{equation}
 \sum_{j \in \{1, \ldots, N-1\}^k} s(\vj, \vy) \of_k(\vj; 1/N)
  = \sum_{j \in \N_+^k} s(\vj, \vy) a_0(\vj) \quad \mbox{} + O \left( \frac{1}{N^{2-\varepsilon}} \right)\,.
\label{eq:prop-exp-cont-to-show-2}
\end{equation}
Since $|s(\vj, \vy)| \le 1$ and $a_0(\vj) > 0$, Equation~\eqref{eq:prop-exp-cont-to-show-2} is implied by the following two lemmata.
\begin{lemma} \label{lem:f-k-conv-finite}
 \stmtlemfkconvfinite
\end{lemma}
\begin{lemma} \label{lem:f-k-conv-infinite}
 \stmtlemfkconvinfinite
\end{lemma}
Lemmata \ref{lem:f-k-conv-finite} and~\ref{lem:f-k-conv-infinite} are proved in the following Subsections
 \ref{sub:lem:f-k-conv-finite} and~\ref{sub:lem:f-k-conv-finite}, respectively.
\qed
\end{proof}

\subsection{Proof of Lemma~\ref{lem:f-k-conv-finite}} \label{sub:lem:f-k-conv-finite}

\newcommand{\Ej}[2]{E_{#1}((j^2)^{#2})}%
In the following, for $\ell \in \N$, let $\Ej{k}{\ell}$ denote
 a sum of monomials of the form $c \cdot j_1^{2 \ell_1} \cdots j_k^{2 \ell_k}$ such that
 $c \in \R$ and $\ell_i \in \{0, \ldots, \ell\}$ for $i \in \{1, \ldots, k\}$ and $\ell_1 + \cdots + \ell_k = \ell$.
For instance, we write $3 j_1^8 - \sqrt{2} j_1^2 j_2^6 = \Ej{2}{4}$.
% $O(F(j_1, \ldots, j_k))$ a function $G(j_1, \ldots, j_k)$ such that there are constants $C, D \in \R$ with
%\stefan{check $O(...)$}
% \[
%  C \cdot F(j_1, \ldots, j_k) \le G(j_1, \ldots, j_k) \le D \cdot F(j_1, \ldots, j_k) \quad \text{for all $j_1, \ldots, j_k > 0$.}
% \]
\begin{lemma} \label{lem:f-k-taylorExpansion}
The function~$\of_k$ has a Taylor expansion
 \[
  \of_k(\vj; u) = a_0(\vj) + a_2(\vj) u^2 + a_4(\vj) u^4 + \cdots
 \]
 with
 \begin{align*}
  a_i(\vj) &= \frac{\Ej{k}{i}}{j_1 \cdots j_k \cdot (j_1^2 + \cdots + j_k^2)^{(i+2)/2}}
  %a_i(\vj) &= \frac{O((j_1^2 + \cdots + j_k^2)^{(i-2)/2})}{j_1 \cdots j_k}
     && \text{for $i=0, 2, \ldots$}
\intertext{More precisely, we have}
  a_0(\vj) & = \frac{2^k}{r \cdot (1-r) \cdot \pi^{k+2} \cdot j_1 \cdots j_k \cdot (j_1^2 + \cdots + j_k^2)}\,,
 \end{align*}
 and for $r \in \rint$,
  all coefficients of the multivariate polynomial in the nominator of $a_4(\vj)$ are negative.
\end{lemma}

\begin{proof}%
\newcommand{\og}{\overline{g}}%
Let $\og(j;u) := g(j,y;u) / \sin(y j \pi)$.
Notice that $1/\og(j;u)$ and $h(j;u)$ have the following Taylor series:
\begin{align*}
  \frac{1}{\og(j;u)} & = c_1 j u + c_3 j^3 u^3 + c_5 j^5 u^5 + \cdots \qquad \text{and} \\
  h(j;u)             & = 1 + d_2 j^2 u^2 + d_4 j^4 u^4 + \cdots
\end{align*}
with $c_1 = \frac{\pi}{2}$ and $d_2 = - r(1-r) \pi^2$.
It follows that we have
\[
 \of_k(\vj; u) = \frac{1}{e_0 + e_2 u^2 + e_4 u^4 + \cdots}
\]
with
\begin{align*}
 e_0 & = \overbrace{\frac{r \cdot (1-r) \cdot \pi^{k+2}}{2^k}}^{c_1^k \cdot (-d_2)} \cdot  j_1 \cdot \ldots \cdot j_k \cdot (j_1^2 + \cdots + j_k^2)  && \text{and} \\
 e_i & = j_1 \cdots j_k \cdot \Ej{k}{(i+2)/2}                                         && \text{for $i = 0, 2, \ldots$}
\end{align*}
Since $e_0 > 0$, the power series $e_0 + e_2 u^2 + e_4 u^4 + \cdots$ can be inverted.
The inversion formula yields
 \[
  \of_k(\vj; u) = a_0 + a_2 u^2 + a_4 u^4 + \cdots
 \]
 with
 \begin{align*}
  a_0 & = \frac{1}{e_0} = \frac{2^k}{r \cdot (1-r) \cdot \pi^{k+2} \cdot j_1 \cdots j_k \cdot (j_1^2 + \cdots + j_k^2)} && \text{and} \\
  a_i & = - a_0 \cdot \sum_{\ell = 0, 2, \ldots, i-2} a_\ell e_{i-\ell} && \text{for $i = 2, 4, \ldots$}
\intertext{It follows by an easy induction that}
  a_i &= \frac{\Ej{k}{i}}{j_1 \cdots j_k \cdot (j_1^2 + \cdots + j_k^2)^{(i+2)/2}}
  %a_i &= \frac{O((j_1^2 + \cdots + j_k^2)^{(i-2)/2})}{j_1 \cdots j_k}
  && \text{for $i=0, 2, \ldots$}
 \end{align*}
Using further values of the Taylor coefficients $c_i, d_i$ from above, a straightforward but tedious computation shows that
\[
 a_4 = \frac{2^{k-4} P(\vj)}{45 \cdot \pi^{k-2} \cdot r \cdot (1-r) \cdot j_1 \cdots j_k \cdot (j_1^2 + \cdots + j_k^2)^3}\,,
\]
where
\begin{align*}
 P(\vj) & = \sum_{1 \le i_1 \le k} -3 j_{i_1}^8 + \sum_{1 \le i_1 < i_2 \le k} -9 (j_{i_1}^6 j_{i_2}^2 + j_{i_1}^2 j_{i_2}^6) + \mbox{} \\
        & \quad \sum_{1 \le i_1 < i_2 \le k} (720 (r(1-r))^2 - 240 r (1-r) + 8) j_{i_1}^4 j_{i_2}^4 + \mbox{} \\
        & \quad \sum_{1 \le i_1 < i_2 < i_3 \le k} (720 (r(1-r))^2 - 300 r (1-r) + 13) (j_{i_1}^4 j_{i_2}^2 j_{i_3}^2 +
                                                                                         j_{i_1}^2 j_{i_2}^4 j_{i_3}^2 +
                                                                                         j_{i_1}^2 j_{i_2}^2 j_{i_3}^4) + \mbox{} \\
        & \quad \sum_{1 \le i_1 < i_2 < i_3 < i_4 \le k} (1440 (r(1-r))^2 - 720 r (1-r) + 60) j_{i_1}^2 j_{i_2}^2 j_{i_3}^2 j_{i_4}^2 \,.
\end{align*}
We have determined the above coefficients of $P(\vj)$ using the computer algebra system Maple.
Now it is straightforward to verify that all coefficients of~$P(\vj)$ are negative, if $\frac{3 - \sqrt{3}}{12} < r \cdot (1-r) \le \frac{1}{4}$.
Those inequalities hold, if $r \in \rint$.
\qed
\end{proof}

The following lemma is used as an induction step in the proof of Lemma~\ref{lem:f-k-conv-finite} below.
\begin{lemma} \label{lem:f-k-induction-step}
 If $k \in \{2, 3, \ldots\}$ and $u > 0$, then
 \[
  \lim_{j_k \to 0} \left( j_k \cdot \of_k(j_1, \ldots, j_k; u) \right) = \frac{2}{\pi} \of_{k-1}(j_1, \ldots, j_{k-1}; u)\,,
 \]
 where the $j_i$ vary over the nonnegative reals.
 Consequently, with the Taylor expansion
  $
   \of_k(\vj; u) = a_{k;0}(\vj) + a_{k;2}(\vj) u^2 + a_{k;4}(\vj) u^4 + \cdots
  $
  from Lemma~\ref{lem:f-k-taylorExpansion} we also have
  \[
   \lim_{j_k \to 0} \left( j_k \cdot a_{k;i}(j_1, \ldots, j_k) \right) = \frac{2}{\pi} a_{k-1;i}(j_1, \ldots, j_{k-1})\,.
  \]
\end{lemma}
\begin{proof}
 As $h(0, u) = 1$, it suffices to show that $\lim_{j_k \to 0} \frac{j_k \sin(j_k \pi u)}{1 - \cos(j_k \pi u)} = \frac{2}{\pi u}$.
 This follows easily from l'Hopital's rule:
 \begin{align*}
  \lim_{j_k \to 0} \frac{j_k \sin(j_k \pi u)}{1 - \cos(j_k \pi u)}
   & = \lim_{j_k \to 0} \frac{\sin(j_k \pi u) + j_k \cos(j_k \pi u) \cdot \pi u}{\sin(j_k \pi u) \cdot \pi u} \\
   & = \frac{1}{\pi u} + \lim_{j_k \to 0} \frac{j_k \cos(j_k \pi u)}{\sin(j_k \pi u)} \\
   & = \frac{1}{\pi u} + \lim_{j_k \to 0} \frac{\cos(j_k \pi u) - j_k \sin(j_k \pi u)}{\cos(j_k \pi u) \cdot \pi u} \\
   & = \frac{1}{\pi u} + \frac{1}{\pi u} = \frac{2}{\pi u}
 \end{align*}
\qed
\end{proof}

Now we can prove Lemma~\ref{lem:f-k-conv-finite} which is restated here.
\begin{qlemma}{\ref{lem:f-k-conv-finite}}
 \stmtlemfkconvfinite
\end{qlemma}
\begin{proof}
 By Lemma~\ref{lem:f-k-taylorExpansion} it is equivalent to prove
 \[
  \sum_{\vj \in \{1, \ldots, N-1\}^k} \left| \frac{a_2(\vj)}{N^2} + \frac{a_4(\vj)}{N^4} + \cdots \right| = O \left(\frac{(\log N)^k}{N^2}\right)\,.
 \]
 Notice that an easy induction shows that
 \[
  \sum_{\vj \in \{1, \ldots, N-1\}^k} \frac{1}{j_1 \cdots j_k}
   = \sum_{j_k=1}^{N-1} \frac{1}{j_k} \sum_{\vj \in \{1, \ldots, N-1\}^{k-1}} \frac{1}{j_1 \cdots j_{k-1}}
   = O((\log N)^k)\,.
 \]
 Hence, with Lemma~\ref{lem:f-k-taylorExpansion} we have
 \[
  \sum_{\vj \in \{1, \ldots, N-1\}^k} \left| \frac{a_2(\vj)}{N^2} \right| = \frac{1}{N^2} \cdot \sum_{\vj \in \{1, \ldots, N-1\}^k} \frac{O(1)}{j_1 \cdots j_k}
   = O \left(\frac{(\log N)^k}{N^2}\right)\,.
 \]
 Similarly, we also have
 \begin{align*}
  \sum_{\vj \in \{1, \ldots, N-1\}^k} \left| \frac{a_4(\vj)}{N^4} \right|
    & = \frac{1}{N^2} \cdot \sum_{\vj \in \{1, \ldots, N-1\}^k} \frac{O(j_1^2 + \cdots + j_k^2)}{N^2 \cdot j_1 \cdots j_k} \\
    & = \frac{1}{N^2} \sum_{\vj \in \{1, \ldots, N-1\}^k} \frac{O(1)}{j_1 \cdots j_k}
      = O \left(\frac{(\log N)^k}{N^2}\right)\,.
 \end{align*}
 Now it suffices to show that, for any fixed $k \in \N_+$,
 \begin{equation*}
  R_k(\vj; 1/N) := \frac{\of_k(\vj; 1/N) - a_0(\vj) - a_2(\vj)/N^2}{a_4(\vj)/N^4} = \frac{a_4(\vj)/N^4 + a_6(\vj)/N^6 + \cdots}{a_4(\vj)/N^4}
 \end{equation*}
 is bounded over all $N \in \N_+$ and all $\vj \in \{1, \ldots, N-1\}^k$,
  because then we also have that
 \[
  \sum_{\vj \in \{1, \ldots, N-1\}^k} \left| \frac{a_4(\vj)}{N^4} + \frac{a_6(\vj)}{N^6} + \cdots \right| = O \left(\frac{(\log N)^k}{N^2}\right)\,.
 \]

 It follows from Lemma~\ref{lem:f-k-taylorExpansion} that $R_k$ depends only on $j_1/N, \ldots, j_k/N$;
  i.e., there is a function
\newcommand{\wR}{\widetilde{R}}%
  $
   \wR_k: [0,1]^k \to \R
  $
  such that $R_k(j_1, \ldots, j_k; 1/N) = \wR_k(j_1/N, \ldots, j_k/N)$.
 Therefore it suffices to take $N=1$ and to show that %\wR_k(j_1, \ldots, j_k) =
  \[
    R_k(\vj; 1) = \frac{\of_k(\vj; 1) - a_0(\vj) - a_2(\vj)}{a_4(\vj)}
  \]
  is bounded over all $\vj \in (0,1]^k$.
 We first argue that $R_k(\vj; 1)$ does not have poles for $\vj \in (0,1]^k$.
 This follows from the facts that (1) the function $\of_k(\vj; 1)$ clearly does not have poles there,
  (2) the coefficients $a_0(\vj)$ and $a_2(\vj)$ do not have poles there by Lemma~\ref{lem:f-k-taylorExpansion},
  and (3) we have $a_4(\vj) < 0$ for $\vj \in (0,1]^k$ by Lemma~\ref{lem:f-k-taylorExpansion}.
 Let $\vj \in [0,1]^k \setminus (0,1]^k$.
 By symmetry, we can assume that there is $\ell \in \{1, \ldots, k\}$ such that
  $j_i \ne 0$ for $i \in \{1, \ldots, \ell-1\}$ and
  $j_i = 0$ for $i \in \{\ell, \ldots, k\}$.
\newcommand{\vzero}{\vec{0}}%
 Denote by $\vzero$ a vector $(0, \ldots, 0)$, whose dimension is clear from the context.
 We need to show that $\lim_{(j_\ell, \ldots, j_k) \to \vzero} R_k(\vj; 1)$ exists and is finite.
 We proceed by induction on~$k$.
 For the induction base, let $k=1$.
 Basic computations show that%
 \[
  R_1(j_1; 1) = \frac{240}{\pi^4 j_1^4} - \frac{60 \sin(j_1 \pi)}{\pi (1 - \cos(j_1 \pi))^2 j_1} \,.
 \]
 One can use l'Hopital's rule to find $\lim_{j_1 \to 0} R_1(j_1;1) = 1$.
 For the induction step, let $k \ge 2$.
 In the following we write $a_{k;i}$ to make the dependence of $a_{i}$ on~$k$ explicit.
 We have:
 \begin{align*}
   & \quad \lim_{(j_\ell, \ldots, j_k) \to \vzero} R_k(\vj; 1) \\
   & = \lim_{(j_\ell, \ldots, j_{k-1}) \to \vzero} \lim_{j_k \to 0} \frac{\of_k(\vj; 1) - a_{k;0}(\vj) - a_{k;2}(\vj)}{a_{k;4}} \\
   & = \lim_{(j_\ell, \ldots, j_{k-1}) \to \vzero}
        \frac{\lim_{j_k \to 0} \left( j_k \cdot ( \of_k(\vj; 1) - a_{k;0}(\vj) - a_{k;2}(\vj) ) \right)}{\lim_{j_k \to 0} \left(j_k \cdot a_{k;4}(\vj)\right)} \\
   & = \lim_{(j_\ell, \ldots, j_{k-1}) \to \vzero}
        \frac{ \of_{k-1}(j_1, \ldots, j_{k-1}; 1) - a_{k-1;0}(j_1, \ldots, j_{k-1}) - a_{k-1;2}(j_1, \ldots, j_{k-1})}
             {a_{k-1;4}(j_1, \ldots, j_{k-1})} \\
   & & \hspace{-30mm}\text{(by Lemma~\ref{lem:f-k-induction-step})} \\
   & = \lim_{(j_\ell, \ldots, j_{k-1}) \to \vzero} R_{k-1}(j_1, \ldots, j_{k-1};1) \,,
 \end{align*}
 where the last limit exists and is finite by induction hypothesis.
\qed
\end{proof}

\subsection{Proof of Lemma~\ref{lem:f-k-conv-infinite}} \label{sub:lem:f-k-conv-infinite}

\begin{lemma} \label{lem:ben-sum-1}
 Let $j_1, \ldots, j_k \in \N_+$.
 Let $\delta, \varepsilon \in \R$ with $\delta > 0$ and $0 \le \varepsilon \le 2 \delta$.
 Then
 \[
  \frac{j_k^\varepsilon}{\left( j_1^2 + \cdots + j_k^2 \right)^{\delta}} \le \frac{1}{\left( j_1^2 + \cdots + j_{k-1}^2 \right)^{\delta - \varepsilon/2}}\,.
 \]
\end{lemma}
\begin{proof}
 Define $r := \sqrt{j_1^2 + \cdots j_{k-1}^2}$.
 Then we need to prove that
 \[
  j_k^\varepsilon \cdot r^{2\delta - \varepsilon} \le \left( j_k^2 + r^2 \right)^\delta \,,
 \]
 or, equivalently,
 \[
  j_k^{\varepsilon/\delta} \cdot r^{2 - \varepsilon/\delta} \le j_k^2 + r^2 \,.
 \]
 %The latter equality follows from a simple case distinction on \mbox{$j_k \le r$} or \mbox{$j_k \ge r$}.
 If $j_k \le r$, then
  $
       j_k^{\varepsilon/\delta} \cdot r^{2 - \varepsilon/\delta}
   \le r^{\varepsilon/\delta} \cdot r^{2 - \varepsilon/\delta}
   =   r^2 \le j_k^2 + r^2
  $.
 The case $j_k \ge r$ is similar.
\qed
\end{proof}

\begin{lemma} \label{lem:ben-sum-2}
 Let $\delta \in \R_+$ and $k \in \N_+$.
 The following series converges:
 \[
  \sum_{(j_1, \ldots, j_k) \in \N_+^k} \frac{1}{j_1 \cdot \ldots \cdot j_k \cdot (j_1^2 + \cdots + j_k^2)^\delta} \,.
 \]
\end{lemma}
\begin{proof}
 The proof is by induction on~$k$.
 The assertion clearly holds for the base case $k=1$.
 For $k \ge 2$ we have:
 \begin{align*}
  & \ \quad \sum_{(j_1, \ldots, j_k) \in \N_+^k} \frac{1}{j_1 \cdot \ldots \cdot j_k \cdot (j_1^2 + \cdots + j_k^2)^\delta} \\
  & = \sum_{j_k \in \N_+} \frac{1}{j_k^{1+\delta}} \sum_{j_1, \ldots, j_{k-1} \in \N_+}
         \frac{j_k^\delta}{j_1 \cdot \ldots \cdot j_{k-1} \cdot \left( j_1^2 + \cdots + j_k^2 \right)^\delta } \\
  & \le \sum_{j_k \in \N_+} \frac{1}{j_k^{1+\delta}} \sum_{j_1, \ldots, j_{k-1} \in \N_+}
         \frac{1}{j_1 \cdot \ldots \cdot j_{k-1} \cdot \left( j_1^2 + \cdots + j_{k-1}^2 \right)^{\delta/2} }
           && \text{(Lemma~\ref{lem:ben-sum-1})} \\
  & = \sum_{j_k \in \N_+} \frac{1}{j_k^{1+\delta}} \cdot C_{\delta/2, k-1} \,,
 \end{align*}
  where $C_{\delta/2, k-1}$ is the constant from the induction hypothesis.
 The series $\sum_{j_k \in \N_+} {j_k^{-(1+\delta)}}$ clearly converges.
\qed
\end{proof}

Now we can prove Lemma~\ref{lem:f-k-conv-infinite}, which is restated here.
\begin{qlemma}{\ref{lem:f-k-conv-infinite}}
 \stmtlemfkconvinfinite
\end{qlemma}
\begin{proof}
 Let
  \[
   R_k(N) :=    \sum_{j_k = N}^\infty \sum_{(j_1, \ldots, j_{k-1}) \in \N_+^{k-1}} \frac{1}{j_1 \cdot \ldots \cdot j_k \cdot (j_1^2 + \cdots + j_k^2)} \,.
  \]
 We have for all $\varepsilon \in (0,2)$:
 \begin{align*}
   R_k(N)
   & = \sum_{j_k = N}^\infty \frac{1}{j_k^{3 - \varepsilon}}
         \sum_{(j_1, \ldots, j_{k-1}) \in \N_+^{k-1}} \frac{j_k^{2-\varepsilon}}{j_1 \cdot \ldots \cdot j_{k-1} \cdot (j_1^2 + \cdots + j_k^2)} \\
   & \le \sum_{j_k = N}^\infty \frac{1}{j_k^{3 - \varepsilon}}
          \sum_{(j_1, \ldots, j_{k-1}) \in \N_+^{k-1}} \frac{1}{j_1 \cdot \ldots \cdot j_{k-1} \cdot (j_1^2 + \cdots + j_{k-1}^2)^{\varepsilon/2}}
            && \text{(Lemma~\ref{lem:ben-sum-1})} \\
   & = \sum_{j_k = N}^\infty \frac{1}{j_k^{3 - \varepsilon}}
          \cdot C_{\varepsilon/2, k-1}
            && \text{(Lemma~\ref{lem:ben-sum-2})} \\
   & \le C_{\varepsilon/2, k-1} \cdot \int_{x = N-1}^\infty \frac{1}{x^{3 - \varepsilon}} \,dx \\
   & = \frac{C_{\varepsilon/2, k-1}}{2-\varepsilon} \cdot (N-1)^{-2 + \varepsilon} = O(N^{-2 + \varepsilon})\,.
 \end{align*}
\qed
\end{proof}

\section{Proof of Theorem~\ref{thm:full}}

Theorem~\ref{thm:full} is restated here.

\begin{qtheorem}{\ref{thm:full}}
 \stmtthmfull
\end{qtheorem}
\begin{proof}
 Recall that Proposition~\ref{prop:Balding} expresses the distribution
 of~$\tim$ in terms of the distributions of one-dimensional random
 walks with absorbing barriers at $0$ and~$N$.  It is well-known that
 such random walks converge, for large~$N$, to an appropriately scaled
 Brownian motion; see \cite[Chapter~XIV]{FellerVol266}.  (In passing
 we remark that the approximation of~$\F{N}$ through $\frac{N^2 \wF}{r
   (1-r)}$ in Proposition~\ref{prop:expectation-continuous} is exactly
 the same approximation; however, there we also establish bounds on
 the rate of convergence, which are not needed here.)  For the
 theorem it suffices to consider $N$ tokens in a Brownian motion
 placed equidistantly on a circle of unit circumference, so that the
 variance of the relative movement of two tokens is $2 \sigma^2 =
 2D/N^2$ per time unit.  We use Balding's analysis~\cite{Balding88} in
 the following.  Let $S(t)$ denote the expected number of tokens at
 time~$t$.  Balding~\cite[p.~740]{Balding88} gives
  \[
   S(t) = 1 + \frac{2 N}{\pi} \sum_{j=1}^\infty \frac{1}{j} \tan \frac{j \pi}{N} e^{-4 \pi^2 j^2 \sigma^2 t} \,.
  \]
\newcommand{\wS}{\widetilde{S}}%
 We have
  \[
   \wS(t) := \lim_{N \to \infty} S(t) = 1 + 2 \sum_{j=1}^\infty e^{-4 \pi^2 j^2 \sigma^2 t} \,.
  \]
 As $S(t) \ge \P{\tim \le t} \cdot 1 + \P{\tim > t} \cdot 3 = 2 \P{\tim > t} + 1$, we have for all $\varepsilon > 0$ that
  \[
   \P{\tim > t} \le \min \left( 1, \frac{S(t) - 1}{2} \right) \le \min \left( 1, (1+\varepsilon) \frac{\wS(t) - 1}{2} \right)
  \]
   for almost all odd~$N$.
\renewcommand{\ts}{t_1}%
\newcommand{\tp}{t_2}%
 With $\ts := \frac{1}{100 \sigma^2}$ we get $\frac{\wS(\ts) - 1}{2} \approx 0.91 < 1$.
 Notice that $\wS(t)$ is decreasing.
 Hence we obtain
  \begin{align*}
   \E \tim & = \int_0^\infty \P{\tim > t}\, dt \le \ts + (1+\varepsilon) \int_{\ts}^\infty \sum_{j=1}^\infty e^{-4 \pi^2 j^2 \sigma^2 t} \, dt \\
             & = \ts + (1+\varepsilon) \sum_{j=1}^\infty \frac{e^{-\pi^2 j^2 / 25}}{4 \pi^2 j^2 \sigma^2}
               \le \frac{0.0285}{\sigma^2} \,,
  \end{align*}
  where the last inequality holds for small~$\varepsilon$.
 The first statement of the theorem follows by setting $\sigma^2 = D/N^2$.
 The second statement follows by noting that with $\tp := \frac{2}{100 \sigma^2}$ we get $\frac{\wS(\tp) - 1}{2} \approx 0.497 < 0.5$.
\qed
\end{proof}

}{}

\end{document}